\documentclass[12pt]{utarticle}
\usepackage[utf8]{inputenc}
\usepackage[12pt]{moresize}
\usepackage{amsmath}
\usepackage{amssymb}
\usepackage{amsthm}
\usepackage{graphicx}
\usepackage{hyperref}
\usepackage{color}
\usepackage[shortlabels]{enumitem}
\usepackage{thmtools}

\newtheorem{theorem}{Theorem}[section]

\newtheorem{proposition}[theorem]{Proposition}
\newtheorem{corollary}[theorem]{Corollary}
\newtheorem{remark}[theorem]{Remark}

\begin{document}

\preprint{ZMP-XX-2016 \\ DESY16-186\\}
                     
\title{Phases of $\mathcal{N}=4$ SYM, S-duality \\ and Nilpotent Cones}
\author{Aswin Balasubramanian\address{{ DESY Theory \\Notkestr. 85\\22607 Hamburg}}\address{{ Department of Mathematics \\ University of Hamburg\\ Bundesstr. 55 \\
20146 Hamburg}}}

\Abstract{
In this note, I describe the space of vacua $\mathcal{V}$ of four dimensional $\mathcal{N}=4$ SYM on $\mathbb{R}^4$ with gauge group a compact simple Lie Group $G$ as a stratified space. On each stratum, the low energy effective field theory is different. This language allows one to make precise the idea of moving in the space of vacua $\mathcal{V}$.  A particular subset of the strata of $\mathcal{N}=4$ SYM can be efficiently described using the theory of sheets in a Lie algebra. For these strata, I study the conjectural action of S-duality. I also indicate some benefits of using such a language for the study of the available space of vacua $\overline{\mathcal{V}}$ on the boundary of GL twisted $\mathcal{N}=4$ SYM on a half-space $\mathbb{R}^3 \times \mathbb{R}^+$. As an application of boundary symmetry breaking, I indicate how a) the Local Nilpotent Cone arises as part of the available symmetry breaking choices on the boundary of the four dimensional theory and b) the Global Nilpotent Cone arises in the theory reduced down to two dimensions on a Riemann Surface $C$. These geometries play a critical role in the Local and Global Geometric Langlands Program(s). }

\maketitle \newpage
\tableofcontents

\section{Introduction}

A natural question that arises when studying S-duality of $\mathcal{N}=4$ SYM is whether symmetry breaking followed by duality is the same as duality followed by symmetry breaking. One motivation for this note is to answer this question in the affirmative for a class of phases of $\mathcal{N}=4$ SYM for arbitrary $G$. More accurately, I show that the nested family of S-duality conjectures that interchange phases of $\mathcal{N}=4$ with gauge group $G$ and the phases of $\mathcal{N}=4$ SYM with gauge group $G^\vee$ are such that symmetry breaking followed by duality lands one in the same phase of the same theory as does duality followed by symmetry breaking. Here, $G$ and $G^\vee$ are such that their root data are related by Langlands/Goddard-Nyuts-Olive duality.

 To achieve this, I use the natural stratification on the moduli space of vacua $\mathcal{V}$. This makes the statement (encapsulated as Proposition \ref{dualityofstrata}) much clearer and amenable to further extensions.  A second motivation is to show that two important Geometric objects, the Local Nilpotent Cone and the Global Nilpotent Cone arise as part of the available choice of symmetry breaking patterns for the GL twisted theory on a half space.

In the study of  non-perturbative dynamics of Supersymmetric Quantum Field Theories, the existence of a moduli space of vacua in the full quantum theory has often played an important role \cite{Seiberg:1994bz}.  One particular example of such a moduli space is that of the Coulomb Branch\cite{seiberg1994electric,seiberg1994monopoles,Seiberg:1997ax}. It is defined to be the moduli space of vacua in which the low energy theory is described by interacting abelian vector multiplets. For the $\mathcal{N}=4$ theory, this happens to be the only vacuum moduli space. The idea of a Coulomb Branch has also played an important role in several closely associated subjects like the study of TQFTs arising from Supersymmetric Quantum Field Theories. For example, the Coulomb Branch plays an important role in the relationship between Donaldson-Witten TQFT and Seiberg-Witten TQFT \cite{Moore:1997pc}. In Section \ref{langlands}, we study its role in another 4d TQFT, namely the one arising from GL twist of $\mathcal{N}=4$ SYM.

In its varied appearances, the fact that the Coulomb Branch is a stratified space is often left implicit. In this note, I want to make the stratification very explicit. For this purpose, it becomes necessary to depart from what has become standard terminology. We will reserve the letter $G$ for the compact group and use $G_\mathbb{C}$ to denote its complexification and assign $\mathfrak{g}$ to be the associated complex Lie algebra.

In the standard terminology that is used in the study of moduli spaces of vacua, no difference is retained between the following two objects
\begin{itemize}
\item The space of VEVs $\mathcal{V}$ for scalar fields which is actually $r$ copies of the space of semi-simple adjoint ($G$-) orbits in the Lie Algebra $\mathfrak{g}$, where $r$ is the number of adjoint scalars.
\item The parameter space of labels for semi-simple adjoint orbits. This is the same as $(\mathfrak{h}/W)^r$, where $r$ is the number of adjoint scalars and $\mathfrak{h}$ is the Cartan subalgebra of $\mathfrak{g}$. This follows from the standard theory of adjoint orbits \cite{collingwood1993nilpotent} applied to semi-simple orbits. Equivalently, this is the space of Chiral Invariants $\mathcal{B}$ whose co-ordinates are the Weyl Invariant polynomials on $\mathfrak{h}$. For example, if $\mathfrak{g}$ is of Cartan type $A_n$, the space $\mathcal{B}$ is parameterized by gauge invariant polynomials built out of the scalars $Tr(\Phi_i^2), Tr(\Phi_i^3), \ldots Tr(\Phi_i^{n+1})$.
\end{itemize}

Both $\mathcal{V}$ and $\mathcal{B}$ are usually called \textit{Coulomb Branches}. For the purposes of this note, it will be important to distinguish between the above two and we will do so from this point. So, let us denote the space of VEVs by $\mathcal{V}$ and the parameter space of Invariant Polynomials as $\mathcal{B}$. I will reserve the word \textit{space of vacua} for the former while using the terms \textit{chiral invariants} or \textit{Coulomb branch} to denote the latter. As we will see, there is value in talking about orbits even in the case where only semi-simple orbits are involved. The relationship between the two is 1:1 when one considers a VEV which is a semi-simple orbit. But, the value in remembering the orbit information is that it is much easier to read the unbroken gauge group on an arbitrary stratum of $\mathcal{V}$ from the orbit than it is from the space of invariant polynomials. In principle, one can retrieve the same information by considering strata in $\mathfrak{h}/W$ and deducing the Weyl group of the unbroken Gauge Group from the stratum in $\mathfrak{h}/W$ that one is in (see Sec \ref{strategy} below).

The strata of $\mathcal{V}$ are described in Section \ref{bulk} using the language of adjoint orbits, especially that of sheets in the complex lie algebra. In order to keep this note short, I have avoided introducing these subjects in great detail. Such an introduction can be found in \cite{collingwood1993nilpotent} \cite{de2009induced} and in \cite{masses} where the theory of sheets is used in a different context. The theory itself was introduced in \cite{borho1979bahnen}.

\section{Bulk Symmetry Breaking and the Strata of $\mathcal{V}$}
\label{bulk}

We can restrict to the case where the gauge group $G$ is compact and simple. When $G$ is semi-simple, the following analysis can be done for each simple factor.

Now, the bosonic part of the action for the $\mathcal{N}=4$ theory on $\mathbb{R}^4$ is
\begin{equation}
\mathcal{S} = \frac{1}{g^2_{YM}} \int dx^4 \text{ Tr} \bigg( \frac{1}{2} F_{\mu \nu} F^{\mu \nu} + D_{\mu} \phi_i D^{\mu} \phi_i + \frac{1}{2} \sum_{i,j=1}^{6} [\phi_i,\phi_j]^2  \bigg)
\end{equation}

where $\mu=0,1,2,3$ are the spacetime indices and $i,j=1 \ldots 6$ are $\mathcal{R}$-symmetry indices. The scalars $\phi_i$ transform in the six dimensional representation of $\mathfrak{so}(6)$. We are interested in the space of Lorentz invariant vacua of the $\mathcal{N}=4$ theory with a $\mathfrak{so}(6)_R$ symmetry. The $\mathfrak{so}(6)_R$ symmetry that emerges in the IR is generically different from the $\mathfrak{so}(6)_R$ of the UV theory with gauge group $G$. The scalar potential in the Lagrangian is quadratic and it admits a family of solutions. Upto gauge invariance, this space of vacua $\mathcal{V}$ can be identified with the space of adjoint orbits in the lie algebra $\mathfrak{g}$ subject to the condition $[\phi_i,\phi_j]^2=0$. Since gauge symmetry is a redundancy in the theory, one identifies an entire adjoint orbit of the VEV $\langle \phi_i \rangle$ to be one single vacuum. Now, the condition $[\phi_i,\phi_j]^2=0$ implies that the adjoint orbit has to be abelian. Usually, this is taken to mean the VEV $\langle \phi_i \rangle$ is valued in a semi-simple orbit. It is technically possible to find non-semi-simple orbits that are abelian but we will not consider them here. The discussion so far was concerning the classical moduli space of vacua. One of the remarkable features of the $\mathcal{N}=4$ theory is that the scalar potential is not quantum corrected \cite{Seiberg:1988ur}. This implies that the moduli space of vacua persists in the quantum theory.

While this note is focused on just the moduli space of vacua of 4d $\mathcal{N}=4$ SYM with gauge group $G$, the broader statements describing the space of vacua as a stratified space apply directly to theories with lower SUSY and the description that follows is likely to be of added value in this larger context. 
 \subsection{General Properties of a Stratification of $\mathcal{V}$}
 
 Before delving into a detailed study of the strata, it is helpful to outline what general properties should be obeyed by a stratification of $\mathcal{V}$. Let $\mathcal{V} = \sqcup_{\alpha_i} \mathcal{V}_{\alpha_i}$ be a decomposition of the $\mathcal{V}$ into strata $\mathcal{V}_{\alpha_i}$ for a partially ordered set $\alpha_i$. Associated to each stratum is an unbroken gauge group $G_{\alpha_i}$ which encodes the massless degrees of freedom  and a collection of W-bosons $\mathbb{W}_{\alpha_i}$ which encodes the part of the Gauge symmetry that was Higgsed. The partial order on the set $\alpha_i$ encodes the possible Higgsing patterns. For example, if a stratum $\mathcal{V}_{\alpha_1}$ can be reached by symmetry breaking from $\mathcal{V}_{\alpha_2}$, then we say that $\alpha_1 < \alpha_2$ in the partial order.
 
 On general physical grounds, it is natural to demand the following properties hold for the decomposition
 
 \begin{enumerate}[(A)]
 \item The low energy effective field theory (EFT) varies smoothly along a fixed stratum $\mathcal{V}_\alpha$. In other words, no new massless gauge fields appear as we move in a fixed stratum. 
 \item If there exists two distinct strata $\mathcal{V}_{\alpha},\mathcal{V}_{\beta}$ such that their unbroken gauge groups $G_\alpha,G_\beta$ can be further Higgsed to a common gauge group $G_\gamma$ corresponding to a stratum $\mathcal{V}_{\gamma}$, then the two different ways of landing on the stratum $V_\gamma$ yield the same set of W-bosons $\mathbb{W}_{\gamma}$ (upto any EM dualities of the low energy theory on $\mathcal{V}_\gamma$).
 \end{enumerate}
 
 When the UV theory has a Lagrangian, a stratification obeying the above properties is provided by the ordinary Higgs mechanism. We will call this the ``natural stratification" of $\mathcal{V}$. In these cases, it also follows that the condition (B) implies condition (A). But, the idea of a moduli space of vacua is often useful in studying theories with no known UV Lagrangians and in such cases, it is not a priori clear as to what the analogues of W-bosons are and if (a version of) condition (B) can automatically imply condition (A). So, it may be useful to think of them as two separate conditions.

\subsection{Alignment of VEVs}

For generic VEVs for the six scalars, the gauge group is broken to $U(1)^{\text{rank}(G)}$. But, we want to study vacua at non-generic values of the VEV. To reach the other interesting vacua, one immediately realizes that the scalar VEVs have to obey certain alignment conditions \footnote{This is similar to the appearance of a condition for vacuum alignment in MSSM where VEVs of the two scalar fields in the theory should be aligned for $U(1)_{EM}$ to remain unbroken after Electroweak symmetry breaking. Such an alignment is energetically favored by the quartic part of the potential for scalars in MSSM. In the present context, there is no quartic term. So, we are just choosing to restrict ourselves to this subset of the vacua.}. So, henceforth, we will concentrate only on those vacua of $\mathcal{N}=4$ that obey some alignment condition that allows some non-abelian gauge group to remain. In fact, we will be even more restrictive. First, it is convenient to combine the six scalars into the three complex scalars. These scalar VEVs are now semi-simple elements in the complex lie algebra. So, from the nature of the scalar potential, we conclude that the space of vacua $\mathcal{V}$ is identified with three copies of the space of semi-simple adjoint orbits. We call a vacuum of $\mathcal{N}=4$ an \textit{aligned vacuum} if the VEVs are all valued in the same sheet of the complex lie algebra (the notion of a sheet is recalled in the Section below). While it would be extremely interesting to study the action of S-duality on the entire set of vacua, we restrict ourselves only to the set of aligned vacua in the rest of the note.

\subsection{The Strategy : Strata of $\mathcal{V}$ and $\mathcal{B}$ using Sheets}
\label{strategy}

Our strategy for describing the aligned vacua of $\mathcal{N}=4$ SYM is the following. Every semi-simple $G$ orbit in the compact real form of $\mathfrak{g}$ can be mapped to a corresponding semi-simple $G_\mathbb{C}$-orbit in the $\mathfrak{g}$\footnote{This property does not extend to other orbits in the complex lie algebra. For example, while there are nilpotent orbits in the complex lie algebra for the $G_\mathbb{C}$ action, there are no analogs in the compact case. }. And both of these can be mapped to strata in the space of Weyl invariant polynomials $(\mathfrak{h}/W)$ by computing the characteristic polynomial. When the characteristic polynomial does not factorize, we are on the top stratum of $\mathfrak{h}/W$. When eigenvalues start to repeat, then the characteristic polynomial would factorize accordingly and one lands on lower strata of $\mathfrak{h}/W$.  This map from semi-simple orbits to strata in $\mathfrak{h}/W$ is many to one. In fact, it is infinitely many to one. If we change the eigenvalues of the semi-simple element (corresponding to changing masses of W-bosons), we would be changing the semi-simple orbit. But, as long as the pattern of eigenvalues remains identical, we would land on the same stratum of $\mathfrak{h}/W$ when we compute the characteristic polynomial. So, looking to label phases by semi-simple orbits themselves would not be a wise idea. Instead, we want to identify the different phases of the physical theory by finding a way to label the strata in $\mathfrak{h}/W$.

 This is where the theory of sheets turns out to be of use. While there are an infinite number of semi-simple orbits in a Lie algebra, there are only a finite number of sheets containing semi-simple orbits. The theory of sheets concerns adjoint orbits (for the $G_\mathbb{C}$ action) in the complex lie algebra. Specifically, a sheet is a union of adjoint orbits of the same dimension. We are concerned here with sheets that have semi-simple elements in them and these are called Dixmier sheets \cite{de2009induced,masses}. Associated to every Dixmier sheet is a Levi subalgebra $\mathfrak{m}$ that centralizes any semi-simple element in the sheet. Under the adjoint quotient map (the eigenvalues map), each Dixmier sheet is mapped to a distinct stratum in $\mathfrak{h}/W$. This parameterization of strata in $\mathfrak{h}/W$ by the centralizing Levi is very convenient for the purpose of the physical question at hand. The unbroken gauge group associated to a fixed stratum in $\mathfrak{h}/W$ is precisely the compact form of the Levi subalgebra associated to the stratum in $\mathfrak{h}/W$ by the theory of sheets. When such a procedure is carried out for every adjoint scalar that has acquired a non-zero VEV, one would go from describing a stratum in $\mathcal{B}$ to a stratum of $\mathcal{V}$. In terms of stratification theory, the resulting stratification of $\mathcal{V}$ (where strata are labeled by the unbroken gauge group) is an example of an isotropy/orbit-type stratification \cite{mather1972stratifications,duistermaat2012lie}.

 One would go in the other direction in the following way.  Let $\mathcal{B}$ be the space of chiral invariants of the phase of theory in which the gauge group $G$ is broken down to $U(1)^{\text{rank}(G)}$. Let $\mathcal{B}_M$ be the space of chiral invariants  of the phase of the theory with unbroken gauge group $M$ where $M$ has a non-trivial semi-simple part. Now, $\mathcal{B}_M$ can be identified with the space of $W(M^{op})$-invariant polynomials, where $W(M^{op})$ is the Weyl group of the semi-simple part of $M^{op}$ where $M^{op}$ is the opposite Levi subalgebra to $M$. Here, ``opposite Levi" is taken to mean that if $\Pi \subset \Sigma$ is the subset of simple of roots corresponding to the Levi $M$, then $M^{op}$ has $\Sigma \setminus \Pi$ as its set of simple roots. And $\mathcal{B}_M$ is a particular stratum in $\mathcal{B}$. Two extreme cases are when $M=U(1)^{rank(G)}$, where we have $\mathcal{B}_M=\mathcal{B}$ and when $M=G$ where we have that $\mathcal{B}_M$ is the ``0'' stratum of $\mathcal{B}$. But, this method is not always transparent when it comes to identifying the unbroken gauge group. For example, when the unbroken gauge group is of type B or C, knowing just the Weyl group of the unbroken Gauge Group is not sufficient to determine the unbroken gauge group. So, it is a lot more convenient to just study the strata of $\mathcal{V}$. 

 As recalled above, the sheets that contain semi-simple orbits are called Dixmier sheets in Lie theory. From the physical perspective, each Dixmier sheet corresponds to a particular choice for what the unbroken Gauge Group is. For example, if the VEV $\langle \phi_i \rangle$ is valued in the principal Dixmier sheet, then the unbroken gauge group is $U(1)^{rank(G)}$. On the other extreme, if the VEV $\langle \phi_i \rangle =0$, we are in the case with no symmetry breaking and the gauge group is $G$. We will call the stratum of $\mathcal{V}$ where the Gauge Group is $U(1)^{rank(G)}$ to be the lowest stratum of $\mathcal{V}$ and the stratum with gauge group $G$ the top stratum of $\mathcal{V}$.

\begin{remark} It must be noted here that the idea of thinking about the space of vevs (upto gauge equivalence) as a stratified space is, by itself, not new. For example, see the review by L. Michel \cite{michel1980symmetry} and Section 9 of Slansky's report \cite{Slansky:1981yr} where this language is discussed in great detail. The theories most directly relevant to  these works are GUT models where there is one preferred vacuum in the quantum theory once the UV theory is fixed. The language of orbits allows for a discussion of the available choices for this vacuum in a way that does not depend too much on the detailed form of the Higgs potential in the UV theory.  The context of the present work is vastly different. One is talking of special Quantum Field Theories in which a moduli space of vacua has survived in the quantum theory (the potential remains quadratic) and it makes sense to \textit{move} (ie change the VEV in the UV theory and thereby change the RG flow from the UV to the IR) in the moduli space of vacua. 
\end{remark}

\section{S-duality}

Fix $M$ to be the subgroup of $G$ that arises as the unbroken gauge group in one of the phases of the Gauge Theory. By ordinary Higgs mechanism, we mean the symmetry breaking deformation of the $\mathcal{N}=4$ theory triggered by a VEV of the lie algebra valued scalar $\langle \phi \rangle$.

\begin{proposition} (``\textbf{Duality of Strata}")   The action of S-duality and that of ordinary Higgs Mechanism triggered by a VEV for the scalar $\langle \phi_i \rangle$ commute. Equivalently, the strata of $\mathcal{V}$ in the theory with gauge group $G$ are in a natural one to one correspondence with the strata  $\mathcal{V}^\vee$ of the dual theory with gauge group $G^\vee$.
\label{dualityofstrata}
\end{proposition}
\begin{proof}
 The symmetry breaking is triggered by a semi-simple element in the lie algebra $\mathfrak{g}$. The unbroken gauge group in this case is a compact group $M$ whose associated complex lie algebra $\mathfrak{m}$ is a Levi subalgebra of the Lie algebra $\mathfrak{g}$. There is a unique Dixmier sheet associated with a choice of a Levi subalgebra. This Dixmier sheet is a particular stratum in the complex lie algebra $\mathfrak{g}$ consisting of a family of semi-simple adjoint orbits that have a same unbroken gauge group. So, this also corresponds to a particular stratum in $\mathcal{V}$.  Now, from the defining relationship between the root systems of $G^\vee$ and $G$ and the fact that Levi subalgebras are classified (upto conjugacy) by proper subsets of the set of simple roots of the root system \cite{borel1949sous,humphreys2008representations}, we have that for every $\mathfrak{m} \subset \mathfrak{g}$, it is true that $\mathfrak{m}^\vee \subset \mathfrak{g}^\vee$. So, the Proposition follows.  \end{proof}

\begin{corollary} (Linguistic) For the phases in question, one can freely alternate between saying ``dual of symmetry breaking" and ``symmetry breaking in the dual theory'' from the point of view of identifying phases. \end{corollary}

It must be noted here that S-duality for $\mathcal{N}=4$ SYM is, by itself, still very much a conjectural statement (see \cite{Goddard:1976qe,Englert:1976ng,Montonen:1977sn} for some of the original conjectures and \cite{Olive:1995sw} for a review of this early period). Highly non-trivial evidence for the S-duality conjecture has been collected over the years in the form of study of monopole moduli spaces \cite{Sen:1994fa,sen1994dyon,gibbons1996sen,hitchin20002} and in the study of partition functions \cite{vafa1994strong,Wu:2008bv}. The former can be interpreted in terms of behaviour of the one particle BPS spectrum on the lowest stratum of $\mathcal{V}$ \cite{seiberg1994electric} while the latter concerns an observable computed in the top stratum of $\mathcal{V}$. For wider reviews with a special focus on the $\mathcal{N}=4$ theory, see \cite{Harvey:1996ur,Weinberg:2006rq}.

More recently, the defining relationship between the two groups $G$ and $G^\vee$ was interpreted as the S-duality between Wilson and 't-Hooft operators in the top stratum of $\mathcal{V}$ \cite{Kapustin:2005py}. One could list numerous other pieces of evidence here. But, the idea of listing a few is to observe that each piece of evidence is collected in a particular stratum of $\mathcal{V}$ (or for a pair of strata, one each on $G$ and $G^\vee$ sides). On the other hand, Proposition \ref{dualityofstrata} concerns how S-duality acts on the entire set of strata themselves.

It is instructive to see Proposition \ref{dualityofstrata} in action by including several examples  below. The case where $G$ is simply laced, the S-duality quite straightforward since each phase is self-dual and what follows is merely a listing of the Levi subalgebras in each case. We still list the Levi subalgebras since it is amusing to see the strata of $\mathcal{V}$ written out. The algorithm that one uses to obtain the Levi subalgebras of $\mathfrak{g}$ is explained in Appendix \ref{countinglevi}. In the non-simply laced cases, there are some non-trivial maps. The case of $G=G_2,F_4$ especially leads to some intriguing maps between the phases of the gauge theory. We label the phases only at the level of complex lie algebras. The corresponding gauge algebras are associated to the compact real forms of the complex lie algebra. Under S-duality, global data about the group like $\pi_1(G)$ and $Z(G)$ play an important role. When speaking of S-duality at the level of groups, this data is critical. But, in everything that follows, we use methods that are sensitive only to lie algebraic data. A finer study at the level of Groups (about which we comment briefly below) would require that we keep track of $\pi_1(G)$ and $Z(G)$. In what follows, when we describe a particular stratum of $\mathcal{V}$ as ``self-dual", it is to be understood as a statement at the level of lie algebras. And the residual Weyl groups symmetries are not indicated but can be inferred from the data in the pair $(\mathfrak{g},\mathfrak{m})$. While $\mathfrak{m}$ details the continuous gauge symmetries that remain, the residual Weyl group symmetry is an indicator of the part of the continuous gauge symmetry is ``broken''. 

\subsection{Type $A$}

In this case, the phases are self-dual. For example, in $G=SU(5)$, the phases are
\begin{center}
\begin{tabular}{|c|c|}
\hline 
$\mathfrak{g}$ & $\mathfrak{m}$ \\ 
\hline 
$A_4$ & $A_4$ \\ 
$A_4$ & $A_3+ \mathfrak{u}(1)$ \\ 
$A_4$ & $A_2+ A_1 + \mathfrak{u}(1)$ \\ 
$A_4$ & $A_1 + A_1 + 2\mathfrak{u}(1)$ \\ 
$A_4$ & $A_1 + 3\mathfrak{u}(1)$ \\ 
$A_4$ & $4\mathfrak{u}(1)$ \\ 
\hline 
\end{tabular} 
\end{center}

\subsection{Type $D$}

Here again, the phases are self-dual. Take as an example $G=SO(8), \mathfrak{g}=D_4$

\begin{center}
\begin{tabular}{|c|c|}
\hline 
$\mathfrak{g}$ & $\mathfrak{m}$ \\ 
\hline 
$D_4$ & $D_4$ \\
$D_4$ & $A_3^{'}+ \mathfrak{u}(1)$  \\
$D_4$ &  $A_3^{''} + \mathfrak{u}(1)$  \\  
$D_4$ & $D_3 + \mathfrak{u}(1)$ \\
$D_4$ & $A_2 + 2\mathfrak{u}(1)$ \\
$D_4$ & $3A_1$ \\
$D_4$ & $2A_1+2\mathfrak{u}(1)$ \\
$D_4$ & $2A_1^{'} + 2\mathfrak{u}(1)$ \\
$D_4$ &  $2A_1^{''} + 2\mathfrak{u}(1)$\\
$D_4$ & $A_1+3\mathfrak{u}(1)$\\
$D_4$ & $4\mathfrak{u}(1)$ \\
\hline 
\end{tabular} 
\end{center}

\subsection{Types $B$,$C$}

Langlands duality exchanges the types $B$ and $C$. Even for the principal Dixmier sheet, the semi-simple orbits of type $B$ and type $C$ are different orbits. And S-duality exchanges the corresponding two phases. Further, Dixmier sheets in Types $B$ and $C$ map to each other exactly under the S-duality map. Viewed purely from the point of the view of the space of Chiral Invariants, S-duality leaves things ``invariant". So, the non-trivial map between Dixmier sheets in type B to Dixmier sheets in type C leading to a highly non-trivial exchange between the phases of the gauge theories would be somewhat obscured if one is merely studying $\mathcal{B}$. For example, we have in $G=SO(9),G^\vee=Sp(10)$, the following phases. 

\begin{center}
\begin{tabular}{|c|c|c|c|}
\hline 
$\mathfrak{g}$ & $\mathfrak{m}$ & $\mathfrak{g}^\vee$ & $\mathfrak{m}^\vee$ \\ 
\hline 
$B_4$ & $B_4$ & $C_4$ & $C_4$  \\
$B_4$ & $A_3 + \mathfrak{u}(1)$ & $C_4$ & $A_3 + \mathfrak{u}(1)$ \\
$B_4$ & $B_3 + \mathfrak{u}(1)$ & $C_4$ & $C_3 + \mathfrak{u}(1)$\\
$B_4$ & $\tilde{A_1} + B_2 + \mathfrak{u}(1)$ & $C_4$ & $\tilde{A_1} + C_2 + \mathfrak{u}(1)$  \\
$B_4$ & $A_2 + A_1 + \mathfrak{u}(1)$ & $C_4$ & $A_2 + A_1 + \mathfrak{u}(1)$ \\
$B_4$ & $A_2 + \mathfrak{u}(1)^2$ & $C_4$ & $A_2 + \mathfrak{u}(1)^2$ \\
$B_4$ & $\tilde{A_1} + \tilde{A_1} + \mathfrak{u}(1)^2$ & $C_4$ & $\tilde{A_1} + \tilde{A_1} + \mathfrak{u}(1)^2$ \\
$B_4$ & $B_2 + \mathfrak{u}(1)^2$ & $C_4$ & $C_2 + \mathfrak{u}(1)^2$ \\
$B_4$& $A_1 + \tilde{A_1} + \mathfrak{u}(1)^2$  & $C_4$ & $A_1 + \tilde{A_1} + \mathfrak{u}(1)^2$  \\
$B_4$ & $\tilde{A_1}$ + $\mathfrak{u}(1)^3$ & $C_4$ & $\tilde{A_1}$ + $\mathfrak{u}(1)^3$  \\
$B_4$& $A_1$ + $\mathfrak{u}(1)^3$ & $C_4$ & $A_1$ + $\mathfrak{u}(1)^3$ \\
$B_4$ & $\mathfrak{u}(1)^4$ & $C_4$ & $\mathfrak{u}(1)^4$ \\
\hline 
\end{tabular} 
\end{center}


\subsection{Types $G_2$, $F_4$}

In \cite{argyres2006s}, it was noted that S-duality for the theory with $G=G_2, F_4$ acts by a non-trivial action on the top stratum of the space of invariants $\mathcal{B}$ (which corresponds to the lowest stratum of $\mathcal{V}$) while the theory on lowest stratum of $\mathcal{B}$ (which corresponds to the top stratum of $\mathcal{V}$) is self-dual. With the stratification at hand, we can now study the action of S-duality on intermediate strata. In both these theories, there is a non-trivial map between the intermediate strata even though we have a ``self-dual'' theory on the top stratum of $\mathcal{V}$. This is due to the fact that the long root and the short root are exchanged under S-duality. 

\begin{center}
\begin{tabular}{|c|c|c|c|}
\hline 
$\mathfrak{g}$ & $\mathfrak{m}$ & $\mathfrak{g}^\vee$ & $\mathfrak{m}^\vee$ \\ 
\hline 
$G_2$ & $G_2$ & $G_2$ & $G_2$  \\
$G_2$ & $A_1 + \mathfrak{u}(1)$ & $G_2$ & $\tilde{A_1} + \mathfrak{u}(1)$  \\
$G_2$ & $\tilde{A_1} + \mathfrak{u}(1)$ & $G_2$ & $A_1 + \mathfrak{u}(1)$  \\
$G_2$ & $2\mathfrak{u}(1)$ & $G_2$ & $2\mathfrak{u}(1)$  \\
\hline
\end{tabular} 
\end{center}

\begin{center}
\begin{tabular}{|c|c|c|c|}
\hline 
$\mathfrak{g}$ & $\mathfrak{m}$ & $\mathfrak{g}^\vee$ & $\mathfrak{m}^\vee$ \\ 
\hline 
$F_4$ & $F_4$ & $F_4$ & $F_4$  \\
$F_4$ & $B_3+\mathfrak{u}(1)$ & $F_4$  & $C_3+\mathfrak{u}(1)$   \\
$F_4$  &  $C_3+\mathfrak{u}(1)$   & $F_4$  & $B_3+\mathfrak{u}(1)$  \\
 $F_4$  & $A_2+\tilde{A}_1+\mathfrak{u}(1)$ &  $F_4$   & $A_2+\tilde{A}_1+\mathfrak{u}(1)$  \\
 $F_4$   & $\tilde{A}_2+A_1+\mathfrak{u}(1)$ &  $F_4$  & $\tilde{A}_2+A_1+\mathfrak{u}(1)$  \\
$F_4$   & $C_2+2\mathfrak{u}(1)$  & $F_4$   & $B_2(\simeq C_2)+2\mathfrak{u}(1)$  \\
 $F_4$  & $A_2+2\mathfrak{u}(1)$  & $F_4$   & $\tilde{A_2}+2\mathfrak{u}(1)$  \\
 $F_4$  & $\tilde{A_2}+2\mathfrak{u}(1)$ &  $F_4$  & $A_2+2\mathfrak{u}(1)$  \\
$F_4$   & $A_1 + \tilde{A_1}+ 2 \mathfrak{u}(1)$ & $F_4$   & $A_1 + \tilde{A_1}+ 2 \mathfrak{u}(1)$  \\
 $F_4$  & $A_1+3\mathfrak{u}(1)$ &  $F_4$  & $\tilde{A_1}+3\mathfrak{u}(1)$  \\
 $F_4$  & $\tilde{A_1}+3\mathfrak{u}(1)$ &  $F_4$  & $A_1+3\mathfrak{u}(1)$  \\
 $F_4$  & $4\mathfrak{u}(1)$ & $F_4$   & $4\mathfrak{u}(1)$  \\
 \hline
\end{tabular} 
\end{center}

\subsection{Types $E_6$,$E_7$,$E_8$}

These are simply laced and it follows immediately that the Dixmier sheets map to each other under S-duality. 

\begin{center}
\begin{tabular}[t]{|c|c|}
\hline 
$\mathfrak{g}$ & $\mathfrak{m}$ \\ 
\hline 
$E_6$& $E_6$  \\
$E_6$ & $D_5+\mathfrak{u}(1)$  \\
$E_6$ & $A_5 + \mathfrak{u}(1)$ \\
$E_6$ & $A_4+A_1+\mathfrak{u}(1)$\\
$E_6$ & $D_4 + 2\mathfrak{u}(1)$\\
$E_6$ & $A_3+A_1+2\mathfrak{u}(1)$\\
$E_6$ & $A_4+2\mathfrak{u}(1)$\\
$E_6$ & $2A_2+A_1+\mathfrak{u}(1)$\\
$E_6$ & $A_3 + 3\mathfrak{u}(1)$\\
$E_6$ & $2A_2 + 2\mathfrak{u}(1)$\\
$E_6$ & $A_2 + 2A_1 + 2\mathfrak{u}(1)$\\
$E_6$ & $A_2+A_1 + 3\mathfrak{u}(1)$\\
$E_6$ & $3A_1+3\mathfrak{u}(1)$\\
$E_6$ & $A_2+ 4\mathfrak{u}(1)$\\
$E_6$ & $2A_1+4 \mathfrak{u}(1)$\\
$E_6$ & $A_1+ 5\mathfrak{u}(1)$\\
$E_6$ & $6\mathfrak{u}(1)$\\ 
\hline
\end{tabular}
\begin{tabular}[t]{|c|c|}
\hline 
$\mathfrak{g}$ & $\mathfrak{m}$ \\ 
\hline 
$E_7$& $E_7$  \\
$E_7$& $E_6+\mathfrak{u}(1)$ \\
$E_7$& $D_6+\mathfrak{u}(1)$ \\
$E_7$& $D_5+A_1+\mathfrak{u}(1)$ \\
$E_7$& $A_6+\mathfrak{u}(1)$ \\
$E_7$& $A_5+A_1+\mathfrak{u}(1)$ \\
$E_7$& $A_4+A_2+\mathfrak{u}(1)$ \\
$E_7$& $A_3+A_2+A_1+\mathfrak{u}(1)$ \\
$E_7$& $D_5+2\mathfrak{u}(1)$ \\
$E_7$& $D_4+A_1+2\mathfrak{u}(1)$ \\
$E_7$& $A_5'+2\mathfrak{u}(1)$ \\
$E_7$& $A_5''+2\mathfrak{u}(1)$ \\
$E_7$& $A_4+A_1+2\mathfrak{u}(1)$ \\
$E_7$& $A_3+A_2+2\mathfrak{u}(1)$ \\
$E_7$& $A_3+2A_1+2\mathfrak{u}(1)$ \\
$E_7$& $2 A_2 + A_1+2\mathfrak{u}(1)$ \\
$E_7$& $A_2+3A_1+2\mathfrak{u}(1)$ \\
$E_7$& $D_4+3\mathfrak{u}(1)$ \\
$E_7$& $A_4+3\mathfrak{u}(1)$ \\
$E_7$& $(A_3+A_1)'+3\mathfrak{u}(1)$ \\
$E_7$& $(A_3+A_1)''+3\mathfrak{u}(1)$ \\
$E_7$& $2A_2+3\mathfrak{u}(1)$ \\
$E_7$& $A_2 + 2A_1+3\mathfrak{u}(1)$ \\
$E_7$&  $4A_1+ 3\mathfrak{u}(1)$ \\
$E_7$&  $A_3+ 4\mathfrak{u}(1)$ \\
$E_7$& $A_2 + A_1+4\mathfrak{u}(1)$ \\
$E_7$& $3A_1'+4 \mathfrak{u}(1)$ \\
$E_7$& $3A_1''+4 \mathfrak{u}(1)$ \\
$E_7$&  $A_2+5 \mathfrak{u}(1)$ \\
$E_7$&  $2A_1+5 \mathfrak{u}(1)$ \\
$E_7$&  $A_1+6 \mathfrak{u}(1)$ \\
$E_7$&  $7\mathfrak{u}(1)$ \\
\hline
\end{tabular}
\end{center}

\begin{center}
\begin{tabular}[t]{|c|c|}
\hline 
$\mathfrak{g}$ & $\mathfrak{m}$ \\ 
\hline 
$E_8$& $E_8$  \\
$E_8$& $E_7+\mathfrak{u}(1)$  \\
$E_8$& $D_7+\mathfrak{u}(1)$  \\
$E_8$& $E_6+A_1+\mathfrak{u}(1)$  \\
$E_8$& $D_5+A_2+\mathfrak{u}(1)$  \\
$E_8$& $A_7+\mathfrak{u}(1)$  \\
$E_8$& $A_6+A_1+\mathfrak{u}(1)$  \\
$E_8$& $A_4+A_3+\mathfrak{u}(1)$  \\
$E_8$& $A_4+A_2+A_1+\mathfrak{u}(1)$  \\
$E_8$& $E_6+2\mathfrak{u}(1)$  \\
$E_8$& $D_6+2\mathfrak{u}(1)$  \\
$E_8$& $D_5+A_1+2\mathfrak{u}(1)$  \\
$E_8$& $D_4+A_2+2\mathfrak{u}(1)$  \\
$E_8$& $A_6+2\mathfrak{u}(1)$  \\
$E_8$& $A_5+A_1+2\mathfrak{u}(1)$  \\
$E_8$& $A_4+A_2+2\mathfrak{u}(1)$  \\
$E_8$& $2A_3+2\mathfrak{u}(1)$  \\
$E_8$& $A_4+2A_1+2\mathfrak{u}(1)$  \\
$E_8$& $A_3+A_2+A_1+2\mathfrak{u}(1)$  \\
$E_8$& $2A_2 + 2A_1+2\mathfrak{u}(1) $  \\
$E_8$& $D_5+3\mathfrak{u}(1)$  \\
\hline
\end{tabular}
\begin{tabular}[t]{|c|c|}
\hline 
$\mathfrak{g}$ & $\mathfrak{m}$ \\ 
\hline 
$E_8$& $D_4+A_1+3\mathfrak{u}(1)$  \\
$E_8$& $A_5+3\mathfrak{u}(1)$  \\
$E_8$& $A_4+A_1+3\mathfrak{u}(1)$  \\
$E_8$& $A_3+A_2+3\mathfrak{u}(1)$  \\
$E_8$& $A_3 + 2A_1+3\mathfrak{u}(1)$  \\
$E_8$& $2A_2 + A_1+3\mathfrak{u}(1)$  \\
$E_8$& $A_2+3A_1+3\mathfrak{u}(1)$  \\
$E_8$& $D_4+4\mathfrak{u}(1)$  \\
$E_8$& $A_4+4\mathfrak{u}(1)$  \\
$E_8$& $A_3+A_1+4\mathfrak{u}(1)$  \\
$E_8$& $2A_2+4\mathfrak{u}(1)$  \\
$E_8$& $A_2+2A_1+4\mathfrak{u}(1)$  \\
$E_8$& $4A_1+4\mathfrak{u}(1)$  \\
$E_8$& $A_3+5\mathfrak{u}(1)$  \\
$E_8$& $A_2+A_1+5\mathfrak{u}(1)$  \\
$E_8$& $3A_1+5\mathfrak{u}(1)$  \\
$E_8$& $A_2+6\mathfrak{u}(1)$  \\
$E_8$& $2A_1+6\mathfrak{u}(1)$  \\
$E_8$& $A_1+7\mathfrak{u}(1)$  \\
$E_8$& $8 \mathfrak{u}(1)$  \\
& \\
\hline
\end{tabular}
\end{center}

Here, it is of interest to note that in the S-duality map between $\mathbb{M}$-regular Surface Operators in $\mathcal{N}=4$ SYM, similar Lie theoretic data (esp the duality between Levi sub-groups in $G$ and $G^\vee$) enters the picture \cite{gukov2006gauge}. 

\subsection{A Group Analog of the Higgs Mechanism}

A conscientious reader would have noted at this point that everything has really been done only at the level of the Lie algebra\footnote{I thank Ashoke Sen for sharing some of his insights in this context. The speculations here (which are my own) are as a partial response to some of his remarks.}. The Gauge theory, on the other hand, depends on the group $G$. One way in which the dependence on the Group can be retrieved is by studying partition functions on compact four manifolds (\cite{vafa1994strong}). Another way would be by detailed study of available extended defect operators. For example, the spectrum of Line operators available in a gauge theory would encode the dependence of the gauge theory on the Group (see, for example \cite{Gaiotto:2010be,Aharony:2013hda}). For the purposes of this note, it would be desirable to speak, directly, of a Higgs mechanism using data that live in the Group (as opposed to passing to Lie algebraic data like $\langle \phi_i \rangle$ in an intermediate step). A Group analog of the Higgs Mechanism is likely to reveal the existence of other phases of $\mathcal{N}=4$ SYM which are not detectable by a study of just the UV Lagrangian since invoking the UV Lagrangian would automatically imply passing to Lie Algebraic data like $\langle \phi_i \rangle$. Here, certain tools that extend the usual Wilson/'t-Hooft criteria that have been developed in \cite{Gukov:2013zka} may be of use. These tools only use Group theoretic data to study properties of phases and one can presumably extend these tools to theories with Coulomb phases (this is also mentioned in \cite{Gukov:2013zka}). It might also be possible to extend the methods in such a way that one can use them to detect the phases using purely Group data. But we do not attempt to present such an extension here.

\section{Geometric Langlands}
\label{langlands}

In the gauge theory approach to Geometric Langlands initiated by Kapustin-Witten, the starting point is the S-duality of $\mathcal{N}=4$ SYM with gauge group $G$ and $\mathcal{N}=4$ SYM with gauge group $G^\vee$. Then, we consider a specific topological twist of the theory. This twist was first studied by Marcus and is now more commonly known as the GL twist \cite{Marcus:1995mq,kapustin2006electric}. The GL twist involves an embedding of the $Spin(4)' \rightarrow Spin(4) \times Spin(6)$ which is such that $\mathbf{4}$ of $Spin(6) \simeq SU(4)$ transforms as $(\mathbf{(2,1)} \oplus \mathbf{(1,2)})$ of $Spin(4)'$. As an embedding of $Spin(4)\simeq SU(2) \times SU(2)$ in $Spin(6) \simeq SU(4)$, it is the diagonal embedding. The six original scalars $\phi_i$ transform in the $\mathbf{6}$ of $Spin(6)$. After GL twist, the symmetry between the six scalars is broken. Four of the scalars, which we take to be $\phi_0,\phi_1,\phi_2,\phi_3$, transform as the components of an adjoint valued one form $\phi=\phi_\mu dx^\mu$. In fact, one can build a complex gauge field $\mathcal{A}= A+ i \phi$ using the gauge field $A_\mu$ and the scalar field $\phi_\mu$. The rest of the scalars combine to form an adjoint valued complex scalar $\sigma  = (\phi_4 + i \phi_5)/\sqrt{2}, \sigma = (\phi_4 - i \phi_5)/\sqrt{2} $. In the rest of our discussion, we will always set $\sigma$ to zero.

 We recall the purely bosonic part of the topological action from \cite{kapustin2006electric},
\begin{eqnarray}
\mathcal{S}_{top} &=& - \frac{1}{g^2_{YM}} \int d^4x \sqrt{g} Tr \bigg( \frac{1}{2}F_{\mu \nu}F^{\mu \nu}  + (D \phi)^2  \bigg)   \\  \nonumber
            &+& \frac{2}{g^2_{YM}} \int d^4x \sqrt{g} Tr \bigg( \frac{1}{2} [\bar{\sigma},\sigma]^2 - D_{\mu} \bar{\sigma} D^\mu \sigma - [\phi_\mu,\sigma][\phi^\mu, \bar{\sigma}]  \bigg) \\ \nonumber
            &-& \bigg( \frac{t-t^{-1}}{g^2_{YM}(t+t^{-1})}\bigg) \int_M Tr(F \wedge F)
\end{eqnarray}

where $t \in \mathbb{CP}^1$ parameterizes the family of supersymmetries $Q_t$ that is important for the GL twist and we have set the $\theta$ angle to zero.
The path integral of the twisted theory localizes on the space of solutions to the following systems of equations on $M_4=\Sigma \times C$ (when $\sigma =0$),
\begin{eqnarray}
\label{kapustinwitten}
(F - \phi \wedge \phi + t D \phi)_+ &=& 0 \\ \nonumber
(F - \phi \wedge \phi - t^{-1} D \phi)_- &=& 0 \\ \nonumber
D^* \phi &=& 0
\end{eqnarray}

The topologically twisted theory continues to have a quantum moduli space of vacua $\mathcal{V}$. These vacua are invariant under the global symmetry $Spin(4)$ of the twisted theory which is a combination of the Lorentz and R-symmetries of the original theory. The considerations of the earlier Section that led us to identify the  space of vacua $\mathcal{V}$ of the theory on $\mathbb{R}^4$ with the space of adjoint semi-simple orbits persist in this twisted case with the caveat that symmetry breaking using the scalar field $\phi$ is now very different from symmetry breaking using the scalar $\sigma$. While both lead to low energy theories with smaller gauge groups, their effect on the 4d TQFT is quite different. In what follows, we actually consider giving VEVs to the scalars $\phi_0,\phi_1,\phi_2,\phi_3$. 

While the S-duality of the physical $\mathcal{N}=4$ theory appears to treat the theory with $G$ and the theory with $G^\vee$ on a somewhat equal footing, the twisted theory and the more refined S-duality map between extended defect operators introduces an asymmetric setup. So, the setup of Kapustin-Witten can be best thought of as leading, in part, to a pair of Geometric Langlands Conjectures, each one of which treats the $G$ and $G^\vee$ side in an asymmetric manner. This is important for us only to the extent that we will now go ahead and pick one of the Geometric Langlands Conjectures that follow and identify the $G^\vee$ side of this conjecture with the Galois side of the correspondence and the $G$ side with the Automorphic side.

In the setup of \cite{kapustin2006electric}, one studies the GL twisted $\mathcal{N}=4$ SYM theory on $\Sigma \times C$ in a regime where the size of $C$ is very small compared to the size of $\Sigma$. For many purposes, the low energy theory in this case can be described as a 2d Sigma Model with target space the Hitchin Moduli space $M_H(G,C)$. The Hitchin moduli space is the space of solutions to the following equations on $C$ \cite{hitchin1987self},
\begin{eqnarray}
\label{hitchin}
F - \phi \wedge \phi = 0 \\ \nonumber
D \phi = 0 \\ \nonumber
D^\star \phi = 0
\end{eqnarray}

This dimensional reduction to the Hitchin sigma model was already studied in \cite{Bershadsky:1995vm,Harvey:1995tg}. The Hitchin moduli space can be thought of as parameterizing particular classes of solutions to the four dimensional equations (\ref{kapustinwitten}) that arise from pullbacks of solutions to (\ref{hitchin}) on $C$ and these constitute the vacua of the effective theory in two dimensions.

Recall now that the Hitchin system has the structure of being a complex integrable system. In particular, it has a map (called the Hitchin fibration)
\begin{equation}
\mu : \mathcal{M}_H \rightarrow \mathcal{H}
\end{equation}
where $\mathcal{H}$ is usually called the Hitchin Base. The Hitchin Base $\mathcal{H}$ is a half-dimensional subspace of $\mathcal{M}_H$ and the generic fibers of $\mu$ are  tori.  

The Global Nilpotent cone $\mathcal{N}_G$ is a geometrical object of critical importance in the Global Geometric Langlands Program and it is defined in a really simple way using the Hitchin fibration $\mu$. We recall its definition here,
\begin{equation}
\mathcal{N}_{G} = \mu^{-1} (0)
\end{equation} 
where $0$ is the ``zero" point of the base $\mathcal{H}$ of Hitchin's fibration. The definition of the Global Nilpotent cone is similar to that of the definition of the Local Nilpotent Cone with the Hitchin fibration playing the role of the Adjoint Quotient Map.

From its definition, it is obvious that the Global Nilpotent Cone is a brane in the Hitchin Sigma Model. In fact, it is a Lagrangian brane \cite{laumon1987analogue,faltings1993stable,ginzburg2001global}. What is not obvious is that there is a way to trace the Global Nilpotent Cone to a four dimensional origin. Below, we will see that strata of the Global Nilpotent Cone do have an origin in specific boundary conditions in four dimensions.

\subsection{Boundary Symmetry Breaking}

We now consider the GL-twisted theory on $\mathbb{R}^3 \times \mathbb{R}^+$ with a supersymmetric boundary condition. It will turn out that the available choices of symmetry breaking at the boundary of the GL twisted theory is larger than those available in the bulk. Let us denote the entire set of boundary vacua to be $\overline{\mathcal{V}}$.

To study $\overline{\mathcal{V}}$, one has to first account for the fact that the asymmetry induced by the GL twist on the six scalars $\phi_i$ is different from that induced by the presence of the boundary which breaks the $\mathcal{R}$-symmetry at the boundary down to $SO(3)_X\times SO(3)_Y$. One of these factors, say $SO(3)_X$, can be embedded in the $Spin(4)$ symmetry of the twisted theory and a choice of such an embedding would constitute the identification of a boundary condition in the twisted theory. So, to describe a boundary condition in the $\mathcal{N}=4$ theory, it is natural to split the scalars into two triplets $\mathbf{X}$ and $\mathbf{Y}$. The triplet $\mathbf{X}$ transforms as the three dimensional representation of $SO(3)_X$ while the triplet $\mathbf{Y}$ transforms as the three dimensional representation of $SO(3)_Y$. 

As we saw in earlier sections, the presence of a quadratic potential for the Higgs fields prevents us from giving a nilpotent VEV in the bulk theory. However, the value that the scalar can take \textit{at} the boundary is not constrained by the form of the scalar potential of the four dimensional field theory. In other words, the scalar can now be valued in any adjoint orbit. In particular, it can be valued in a nilpotent orbit. So, $\overline{\mathcal{V}}$ includes nilpotent orbits while $\mathcal{V}$ does not. To be more precise, one can build a single complex scalar from two of the three scalars that comprise the $SO(3)_X$ triplet $\mathbf{X}$ and allow this complex scalar to take arbitrary nilpotent values on the boundary of $\mathbb{R}^3 \times \mathbb{R}^+$. It is this complex scalar that becomes a part of the Hitchin system on $C$.

The chiral operators that parameterize $\mathcal{B}$ are not suited to detect these additional choices precisely because they are all identically zero for every nilpotent orbit. In more compact terminology, the entire Local Nilpotent Cone is the fiber over ``0'' of the adjoint quotient map (i.e eigenvalues map). The image of the adjoint quotient map is the ring of invariant polynomials and this space can be viewed as $\mathfrak{h}/W$ and the ``0'' of the space is the case where all the invariant polynomials are zero.

So, we takeaway that there is a valid topological boundary condition in the four dimensional theory in which the scalar that enters the Hitchin system can take nilpotent values at the boundary.  Now, we compactify the theory by formulating it on $\Sigma \times C$, where $\Sigma=\mathbb{R}^+ \times \mathbb{R}$. We dimensionally reduce on $C$ and the effective theory in two dimensions can be described using the Sigma model on $\Sigma$ with target the Hitchin moduli space $\mathcal{M}_H(G,C)$. At the boundary of $\Sigma$, one has to prescribe a boundary condition for the Sigma model. In the present case, we are interested in compactifying the boundary condition in four dimensions that involves giving nilpotent VEVs to the scalar on the boundary. In the dimensional reduction, this boundary condition (which can be thought of as a three dimensional modification of the four dimensional theory) is reduced on $C$ leading ultimately to a boundary condition in two dimensions. Since the scalar in four dimensions was fixed to be a particular nilpotent orbit, the reduction forces the scalar to take that nilpotent value globally over $C$. In other words, the resulting brane in the Hitchin Sigma Model in two dimensions is the subspace of the solutions to Hitchin's equations on $C$ with Higgs field taking values in a particular nilpotent orbit\footnote{Here, we are making an assumption that an analog of the vanishing theorems that would allow one to pullback solutions to the Hitchin System on $C$ to honest solutions of the four dimensional equations holds when the scalar takes nilpotent values at the boundary. For the case where the scalar is valued in the principal Dixmier sheet, this argument is presented in \cite{kapustin2006electric}}. This is precisely a stratum of the Global Nilpotent Cone in $M_H(G,C)$. If one varies the four dimensional boundary condition and allows the scalar field to take values in arbitrary nilpotent orbits, then one can get all the strata of the Global Nilpotent Cone as part of  boundary condition(s) in two dimensions. In its origin from four dimensions, when the scalar field is fixed to be a particular nilpotent orbit, the resulting brane in the two dimensional theory is very similar to the brane(s) supported on the fiber at a smooth point of the base $\mathcal{H}$ in that all of these branes can be traced to the choice of a single boundary vacuum of the four dimensional theory. The Electric and Magnetic Eigenbranes of Kapustin-Witten are branes of this type \cite{kapustin2006electric, Witten:2015dta}. It is also useful to note that the boundary conditions that lead to strata of the Global Nilpotent Cone in the two dimensional theory are available on both the Automorphic ($G$) and the Galois ($G^\vee$) sides of the correspondence.

In this framework to study geometric Langlands (further enlarged and developed to more detail by Witten and collaborators), the Local Geometric Langlands Program is expected to be part of the higher dimensional theory while the Global Geometric Langlands is obtained in the dimensionally reduced theory. In the mathematical approaches to the program, both the Nilpotent Cones enter the picture in crucial ways. I will not try to detail this in any way here but refer the reader to  \cite{laumon1987analogue}, \cite{arinkin2015singular} for the Global Nilpotent Cone, to \cite{ginsburg1989admissible,mirkovic2004character} for the Local Nilpotent Cone and to \cite{ben2016betti} for a broader overview. In \cite{ben2016betti}, a Betti version of the Geometric Langlands Program has been introduced and this program is placed in  a mathematical framework that is intrinsically four dimensional in its origins. The appearances of Local and Global Nilpotent Cones presented here are similar to how they appear in \cite{ben2016betti} and further work incorporating these symmetry breaking choices into the setup of \cite{kapustin2006electric}, \cite{gaiotto2008s,Gaiotto:2016hvd} should make the connection precise. We note here that the Global Nilpotent Cone on the Automorphic side has also been discussed in the context of the Conformal Field Theory approaches to Geometric Langlands \cite{frenkel2007lectures,teschner2011quantization}. 

Finally, it must be noted that the strata of the Local Nilpotent Cone can arise in other ways on the boundary of the four dimensional theory. For example, it can enter through the Nahm pole boundary condition or via coupling to three dimensional $\mathcal{N}=4$ SCFTs which have vacuum moduli spaces given by strata in the Local Nilpotent Cone(s). These two ingredients did not play any role in this note. For a complete story, one however needs to take all of these possibilities in the four dimensional $\mathcal{N}=4$ theory into account. One can also further enrich the story by including ramifications at punctures on $C$ along the lines of \cite{gukov2006gauge,Witten:2007td}.

\textit{Acknowledgements} : I would like to thank Jacques Distler for detailed comments on a draft version of this note. I would also like to thank David Ben-Zvi, Ashoke Sen and Joerg Teschner for patient explanations.

\appendix

\section{Counting Levi Subalgebras}
\label{countinglevi}
Let $\mathfrak{g}$ be a complex semi-simple Lie Algebras and $\mathbf{R}$ the associated root system. Classifying the Levi subalgebras of $\mathfrak{g}$ can be turned into a combinatorial problem involving $\mathbf{R}$ by a two step procedure. The first step is

\begin{proposition}  (this is Prop 6.2 of \cite{bala1976classes})  Let $\mathfrak{m}_1$ and $\mathfrak{m}_2$ be two Levi subalgebras in $\mathfrak{g}$ and $\mathbf{R}_1$ and $\mathbf{R}_2$ their associated root systems. Then, they are conjugate under $G$ if and only if $\mathbf{R}_1$, $\mathbf{R}_2$ are conjugate under the Weyl group. 
\label{propa1}
\end{proposition}

In the mathematical literature, it has become standard terminology to name the root systems arising from Levi subalgebras a \textit{parabolic} subsystem of $\mathbf{R}$ \cite{bala1976classes}. This is potentially confusing since at no point in the discussion did we actually use a parabolic subalgebra of $\mathfrak{g}$. But, we will adopt this terminology. Now, the second step in classifying Levi subalgebras concerns these parabolic root systems.
\begin{proposition} (this is Prop 6.3 of \cite{bala1976classes}) Two parabolic subsystems $\mathbf{R}_1$ and $\mathbf{R}_2$ are conjugate under the Weyl group $W(\mathfrak{g})$ if and only if their Dynkin diagrams are identical except in the following cases 
\label{propa2}
\begin{enumerate}
\item If $\mathbf{R}$ is of type $D_{2n} (n \geq 2)$, then we have two non-conjugate systems of type $A_{i_1} + A_{i_2} + \cdots $ where $i_1,i_2,i_3 \ldots i_k$ are odd integers that satisfy $(i_1 +1) + (i_2+1) + \ldots + (i_k+1) = 2n$,
\item If $\mathbf{R}$ is of type $E_7$, then there are two non-conjugate systems of type $A_3+A_1,A_5$ and $3A_1$.
\end{enumerate}
\end{proposition}

So, it follows that in a vast number of cases, the problem of finding Levi subalgebras $\mathfrak{g}$ reduces to the problem of finding distinct sub-diagrams of the Dynkin diagram. The exceptions to this are listed above in Proposition \ref{propa2}. In these case, to denote the Levi subalgebra completely, primes are introduced to distinguish the two distinct Levis that share the same sub-diagram of the Dynkin diagram. This is the reason for the appearance of primes in the labels for certain Levi subalgebras in the case of $\mathfrak{g}=D_4, E_7$ in the main text.
\bibliography{coulombbranch}
\bibliographystyle{alphanum} 
\end{document}